\documentclass[twocolumn,10pt]{article}
\usepackage[utf8]{inputenc}
\usepackage{amsmath,amssymb}
\usepackage{multirow}
\everymath{\mathtt{\xdef\tmp{\fam\the\fam\relax}\aftergroup\tmp}}

\usepackage[textwidth=18.35cm,textheight=24cm]{geometry} 
\usepackage{balance} 

\usepackage{amsthm}
\newcommand\Sym{\mathit{Sym}}

\newtheorem{thm}{Theorem}
\newtheorem{cor}[thm]{Corollary}
\newtheorem{lem}[thm]{Lemma}
\newtheorem{prop}[thm]{Proposition}

\newtheorem{defn}[thm]{Definition}

\newcommand{\manaut}[1]{}

\usepackage[numbers]{natbib}
\newcommand{\textcite}[1]{\citeauthor*{#1}~\cite{#1}}
\newcommand{\printbibliography}{
\bibliographystyle{plainnat} 
\bibliography{LongSteepestRefs}
}

\usepackage{authblk}

\title{Steepest ascent can be exponential in bounded treewidth problems}

\author[1]{David A. Cohen}
\author[2]{Martin C. Cooper}
\author[3,4]{Artem Kaznatcheev}
\author[5]{Mark Wallace}
\affil[1]{Department of Computer Science, Royal Holloway University of London, Egham, UK}
\affil[2]{IRIT, University of Toulouse III, Toulouse, France}
\affil[3]{Department of Computer Science, University of Oxford, Oxford, UK}
\affil[4]{Department of Translational Hematology \& Oncology Research, Cleveland Clinic, Cleveland, USA}
\affil[5]{Faculty of Information Technology, Monash University, Melbourne, Australia}

\date{}

\begin{document}
\maketitle

\begin{abstract}
We investigate the complexity of local search based on steepest ascent.
We show that even when all variables have domains of size two and the underlying constraint graph  of variable interactions has bounded treewidth (in our construction, treewidth 7), there are fitness landscapes for which an exponential number of steps may be required to reach a local optimum.
This is an improvement on prior recursive constructions of long steepest ascents, which we prove to need constraint graphs of unbounded treewidth.
\end{abstract}

\section{Introduction}

We are interested in the long-term behaviour of steepest-ascent local search in optimisation problems over finite domains.
Local search is an important tool in solving optimisation problems when exhaustive search algorithms, such as branch-and-bound, are not a practical possibility due to the size of the search space.

It is therefore important to gain understanding of the behaviour of local search and, in particular, the number of steps required to reach a local optimum. If local search can continue for a number of steps which is exponential in the number of variables, then in practice we can consider that it may never reach a stable state within a reasonable timescale.
This failure to equilibrate is especially important to understand for commonly used greedy local search heuristics like steepest ascent.

Algorithmically, steepest ascent proceeds by a sequence of steps towards a local maximum of the objective function. 
At each step, a best move is chosen to transform the current state into a better neighbouring state. 
In this paper we restrict our attention to the simplest case in which all variables are Boolean and the set of possible moves from a given state is just the flipping of one variable. 

Throughout the paper, we will refer to the objective function being maximized as the fitness function.
We choose this terminology due to the central role that local search has played in the study of biological evolution.
This has been the case ever since \manaut{Wright}\textcite{W32} introduced the idea of viewing the evolution of populations of organisms as a local search process over a space of possible genotypes with associated fitness values that became known as a ``fitness landscape."
For over 80 years since Wright's introduction of fitness landscapes, it has been assumed that if the fitness function is not time-varying (i.e. if we have a static fitness landscape) then typical populations will quickly evolve to a local fitness peak.
Recently, \manaut{Kaznatcheev}\textcite{artemGenetics} has challenged this assumption by showing biologically reasonable families of landscapes where a local peak cannot be found in time polynomial in the number of genes (variables).
On such ``hard" fitness landscapes, we can consider local peaks as practically unreachable.
This situation can be seen as open-ended evolution~\cite{OEE}.

In the case of biological evolution, the particular local search algorithm that is implemented can depend on the details of the particular population structure, genetic architecture, and mutation rate~\cite{artemGenetics}.
However, in the limit of large populations with small mutations rates, the evolutionary dynamics are often modeled as steepest ascent.
Thus, rigorously analyzing steepest ascent can help us better understand biological evolution.
Of course, biology is not the only domain where local search matters.
Similar models are used in fields like business operation \& innovation theory~\cite{orgBehavOrig,orgBehavNetwork}, physics, and economics~\cite{R09}. 

\subsection{Soft constraints and fitness functions}

A discrete optimisation problem can be formalised as a set of functions on subsets of its variables.  The value of the objective function is the sum of the values of these functions. 
In the generic framework known as the VCSP (Valued Constraint Satisfaction Problem), each function is known as a ``soft constraint" on its variables~\cite{soft}.
The ``scope” of a soft constraint is its set of variables; each combination of values assigned to these variables  is mapped to a value by the soft constraint. 
The ``arity” of a constraint is the number of variables in its scope.  
The ``degree” of a variable is the number of other variables that occur with it in some constraint scope. 
In the formulation of the VCSP considered in this paper the objective function is to be maximised. 

Biologically a combination of gene expressions underlies a phenotypic trait, and a combination of its phenotypic traits underlie the fitness.  
Put very simply, fitness is the sum of the values of its phenotypic traits. 
We can assign this same value to the combination of gene expressions which underly the phenotype traits.  
The number of genes involved in a single phenotypic trait is assumed to be reasonably small – say less than some fixed number $K+1$ (the amount of epistasis), and any one gene is involved in a reasonably small number of phenotypic traits – say some fixed number $P$ (the amount of pleiotropy).

A VCSP is a model of fitness if we associate a variable with each gene, which takes the value 1 if the gene is expressed.  
Each soft constraint over a set of variables, models the set of genes underlying a phenotypic trait, and the value of the soft constraint represents the value of the phenotypic trait.  
The bound $K+1$ on the number of genes underlying a phenotypic trait is modelled by a bound on the arity of the soft constraints.   
The bound $P$ on the number of different phenotypes involving one gene is modelled by a bound on the degree of any variable in the VCSP.  

In this paper we only consider functions expressible as a set of soft constraints with bounded arity (our examples all have arity at most eight).  
The \emph{constraint graph} is the graph whose vertices are the variables and with an edge between two variables if there is a soft constraint between them (i.e. they both belong to the scope of some soft constraint).
Given a constraint graph $G$, and a variable index $i$, we will let $d_G(i)$ be the \emph{degree of $i$ in $G$}: i.e. the total number of other variables that co-occur with $x_i$ in the scope of some constraint in the VCSP.

\subsection{Bounded treewidth and local search}

We are particularly interested in the behaviour of local search when the constraint graph has bounded treewidth. 
When treewidth is bounded, it is well known that it is possible to find a global optimum of a VCSP in polynomial time~\cite{BB}, 
but this does not inform us on the complexity of local search even for the simpler problem of reaching a local maximum. 
A major drawback of local search as a technique for solving optimisation problems is the possibility of there being an exponential number of local maxima whose value may even be exponentially smaller than the value of the global maximum. 
It is not difficult to show that this phenomenon can occur when treewidth is bounded. 
Consider, for example, the global objective function

\begin{equation}
    F(x) \ =\ \sum_{i=1}^{N/2} f(x_{2i-1},x_{2i})
\end{equation}

\noindent where 

\begin{equation}
f(a,b)   = \begin{cases}
    1 & \text{ if } \ a = b = 0 \\
    \alpha  & \text{ if } \ a = b = 1 \\
    0  & \text{ otherwise } 
\end{cases}
\end{equation}

\noindent where $x=\langle x_1,\ldots,x_N \rangle \in \{0,1\}^N$, and $N$ is even.
Each of the tuples consisting of repeated bits (i.e. $x_{2i-1}=x_{2i}$ for $i=1,\ldots,N/2$) is a local maximum. 
Furthermore, the ratio between the value of the global maximum and the value of the worst local maximum is $\alpha$ and hence is exponential if $\alpha=2^N$. 
The constraint graph is a forest (consisting of $N/2$ unconnected edges) and hence has treewidth 1. 

Thus bounded treewidth does not alleviate the phenomenon of multi-modality. 
The possibility of exponential complexity of  local search to reach a local maximum, which is the topic of this paper, is an independent phenomenon which, as we have seen, is of interest in genetics as well as computer science.

\subsection{Summary}

We show that local search following a steepest ascent can take an exponential number of steps to reach a local optimum even when (a) the variables are Boolean, (b) the problem is formalised using bounded arity soft constraints, and (c) the constraint graph has bounded treewidth.
Previously, classes of objective functions have been described for which steepest ascent may require an exponential number of variable flips to reach a local maximum~\cite{longpath,artemGenetics} but the fitness functions were defined recursively and not expressed as the sum of bounded-arity functions. 
We first show in the next section that these examples cannot be expressed by soft constraints whose constraint graph has bounded treewidth.
In related work, \manaut{Kaznatcheev, Cohen, and Jeavons}\textcite{artemCP} have studied the complementary problem of guaranteeing short paths to a local optimum for local search which performs an arbitrary improving flip at each step rather than a steepest ascent.
In the process, they have described some bounded treewidth examples where some exponentially long improving paths exist but these are not the paths followed by steepest ascent.

\section{Earlier recursive construction of hard fitness landscapes}

We are not the first to consider long steepest-ascent walks in fitness landscapes.
Two constructions from the literature~\cite{longpath, artemGenetics} of fitness landscapes on $n$ variables where steepest ascent requires of the order of $2^n$ steps are particularly relevant.
However, both these landscapes were specified recursively and not by concrete VCSP instances.
In this section, we give a simpler presentation of the recursive definition of winding landscapes (Definition~\ref{def:winding}) that generalises both the construction in \cite{longpath} and \cite{artemGenetics}.
We then prove that if these winding landscapes were implemented by a VCSP then the fitness graph would have unbounded treewidth (Corollary~\ref{cor:unboundedTreewidth}).
This will motivate consideration of a different form of fitness landscape in Section~\ref{sec:exponential}.

\subsection{Recursive construction}
\label{sec:recursiveHard}

Given a fitness landscape on $\{0,1\}^{2n}$, we refer to the hypercube $\{0,1\}^{2k}$ on the first $2k$ variables as a sub-cube.

\begin{defn}
Suppose we are given a sequence of fittest steps $s^+_k > 0$ and barrier steps $s^-_k$ for $0 \leq k \leq n$ satisfying the conditions  $s^+_k > s^-_{k + 1}$,  $s^-_k < s^+_k < s^+_{k+1}$.  
Consider the recursive construction of the functions $f^k \{0,1\}^{2k} \rightarrow \mathbb{R}$ 
with sub-cube optima $x_k^* = 0^{2(k - 1)}11$ for $0 \leq k \leq n$ (where $x_0^* = \epsilon$, the empty string):

\begin{align}
    f^0(\epsilon) & = 0 \\
    f^{k + 1}(xab) & = \begin{cases}
f^k(x) \hfill  \text{if } a = b = 0 \\
f^k(x) + s_{k + 1}^- \hfill \text{if } a \neq b \text{ \& } x \neq x_k^* \\
f^k(x_k^*) + s_{k + 1}^- \hfill  \text{if } a = 0, b = 1 \text{ \& } x = x_k^* \\
f^k(x_k^*) + s_{k + 1}^+ \hfill \text{if } a = 1, b = 0 \text{ \& } x = x_k^* \\
f^k(x\oplus x_k^*) + f^k(x_k^*) + 2s_{k + 1}^+ \hfill  \text{if } a = b = 1
\end{cases}
\end{align}

\noindent We define the ($n$th) \emph{winding landscape} as $f = f^n$.
\label{def:winding}
\end{defn}

\noindent Here $x \oplus x_k^*$ means the bit-wise XOR of $x$ and $x_k^*$ (i.e. $[x \oplus x_k^*]_i = [x]_i + [x_k^*]_i \mod 2$).

Based on this definition, if the steepest-ascent path takes $T_k$ steps to reach $x_k^*$ from $0^{2k}$ in $f^k$ then steepest-ascent will take the following path from $0^{2k}$ to $x_{k + 1}^*$ in $f^{k + 1}$ (see Appendix C.2 of \cite{artemGenetics} for the proof):

\begin{equation}
0^{2(k + 1)} \rightarrow^{T_k} x_k^*00 \rightarrow x_k^*10 \rightarrow x_k^*11 \rightarrow^{T_k} 0^{2k}11 = x_{k + 1}^*
\end{equation}

\noindent which has a length of $T_{k + 1} = 2T_k + 2$ steps.
Solving this recurrence equation with $T_0 = 0$, we get $T_n = 2^{n + 1} - 2$.

On the one hand, if $s^-_k \leq 0$ for all $0 \leq k \leq n$ then Definition~\ref{def:winding} implements \manaut{Horn, Goldberg, and Deb}\textcite{longpath}'s \emph{Root2path} construction and the long path from $0^{2n}$ is not only the steepest ascent but also the only ascent.
On the other hand, if $s^-_k > 0$ for all $0 \leq k \leq n$ then Definition~\ref{def:winding} implements \manaut{Kaznatcheev}\textcite{artemGenetics}'s winding semismooth fitness landscape that has no reciprocal sign epistasis but still maintains the long path from $0^{2n}$ as the steepest ascent.
This semismooth case is of interest because a short ascent exists from each variable assignment to the unique fitness peak but steepest ascent does not find this short ascent.

\subsection{Lower bounds on the constraint-graph}

The winding fitness landscape has the dramatic property of having drastic changes in the direction and magnitude of the gradient of the objective function between points which are very close to each other.
Consider the sub-cube spanned by the first $2(k+1)$-variables. 
The sub-cube fitness maximum is at $x^*_k= 0^{2k}11$.
If $s^-_{k + 1} > 0$ (as in~\cite{artemGenetics}) then the sub-cube fitness minimum is only Hamming-distance $2$ away at $0^{2(k + 1)}$: so the flows (i.e. differences in the objective function from the current point to its neighbouring points) change from all positive to all negative in just $2$ steps.
From this we can prove that the total scope size of any constraint
graph implementing this fitness landscape must be high.
The case in which we do not necessarily have $s^-_{k + 1} > 0$ (as in \cite{longpath}) is slightly more complicated since $0^{2(k + 1)}$ is no longer a fitness minimum and mostly has negative flows. 
However, these negative flows have small magnitude compared to the very large magnitude negative flows at $x^*_k$, so a similar argument can be used.
We formalise this argument below.

For convenience, let $||x||_0$ be the number of non-zero entries in $x$ -- also known as the Hamming weight or the zero-`norm' of the vector $x$.
And let us use $x[i \rightarrow b]$ to mean a bit-string that is the same as $x$ at every bit, except the $i$-th bit is set to $b$.
Or, in symbols: $\forall j \neq i \quad [x[i \rightarrow b]]_j = [x]_j$ and $[x[i \rightarrow b]]_i = b$.
This allows us to define the \emph{gradient} $\nabla f$ of a fitness function $f$ entry-wise as $[\nabla f(x)]_i = f(x[i \rightarrow 1]) - f(x[i \rightarrow 0])$ to state our degree lower-bound lemma:

\begin{lem}
Given a fitness function $f$ implemented by a VCSP with constraint graph $G$ and any two distinct variable assignments $x$ and $y$ that differ on a set of variables $S$, we have that the total degree $d_G(S) = \sum_{i \in S}d_G(i)$ of $S$ in $G$ is lower-bounded by the change in flow: $d_G(S) \geq ||\nabla fx - \nabla fy ||_0$.
\label{lem:lbD}
\end{lem}

\begin{proof}
If we look at a variable at position $i$ and compare $\nabla f(x[i \rightarrow 1])$ to $\nabla f(x[i \rightarrow 0])$ then any differences in the gradient must have been due solely to the change in variable $x_i$.
Thus, given any position $j$ such that $[\nabla f(x[i \rightarrow 1])]_j \neq [\nabla f(x[i \rightarrow 0])]_j$ there must be a constraint that has both $i$ and $j$ (and maybe others) in its scope.
Thus, by looking at the number of non-zero entries in $\nabla f(x[i \rightarrow 1]) - \nabla f(x[i \rightarrow 0])$, 
we get a lower bound on the number of other variables with which each variable $i$ co-occurs in a constraint.

This reasoning can be extended over paths between non-adjacent states.
Suppose we have two states $x_1$ and $x_t$ at Hamming distance $t$ from each other.
Let $x_1x_2...x_{t - 1}x_k$ be any shortest path between them with the bits flipped at each step given by $i_1, i_2, ... i_{t - 1}$.
Notice the following:

\begin{align}
||\nabla fx_1 - \nabla fx_k||_0 &= ||(\nabla fx_1 - \nabla fx_2) + (\nabla fx_2 - \nabla fx_3) + \nonumber \\
& \cdots + (\nabla fx_{t - 1} - \nabla fx_t)||_0 \\
&\leq ||\nabla fx_1 - \nabla fx_2||_0 + ||\nabla fx_2 - \nabla fx_3||_0 + \nonumber \\
& \cdots + ||\nabla fx_{t - 1} - \nabla fx_t)||_0 \label{eq:lowerbound}
\end{align}

\noindent In words: given two states $x_1$ and $x_k$ that differ at a set of variables $S$, the total number of variables that the variables in $S$ co-occur with is lower-bounded by $||\nabla fx_1 - \nabla fx_k||_0$.
\end{proof}

Now, we can apply this lower bound technique to the winding fitness landscape.

\begin{prop}
If the winding fitness landscape $f$ on $2n$ variables from Definition~\ref{def:winding} is implemented by a VCSP with constraint graph $G$ then  $d_G(2k + 1) + d_G(2k + 2) \geq k$ for each $0 \leq k < n$.
\label{prop:highDegree}
\end{prop}

\begin{proof}
Let us look at the gradients at the path's starting point:

\begin{equation}
    \nabla f(0^{2n}) = [s^+_1, s^-_1, s^-_2, s^-_2, ..., s^-_i, s^-_i, ..., s^-_n, s^-_n]
    \label{eq:bot}
\end{equation}

\noindent and for each $1 \leq k \leq n$, look at the gradients at subcube peaks $\nabla f(0^{2(k - 1)}(11)0^{2(n - k)})$: they have a slightly more complicated form, so we define them point-wise for $i \in [1,n]$, $b \in \{0,1\}$, and $x = 0^{2(k - 1)}(11)0^{2(n - k)}$:

\begin{equation}
    [\nabla f(x)]_{2i - b} = \begin{cases}
    -s^+_i + b(s^-_i - s^+_i) & \text{if } i < k \\
    f^k(x^*_k) - s^-_k & \text{if } i = k \\
    s^+_{k + 1} & \text{if } i = k + 1 \text{ \& } b = 1 \\
    s^-_i & \text{if } i > k + b
    \end{cases}
    \label{eq:top}
\end{equation}

\noindent Looking at the odd entries lower than $2k$ (i.e. $i < k, b = 1$ in equation~\ref{eq:top}),
we have $$[\nabla f(0^{2(k - 1)}(11)0^{2(n - k)}) - \nabla f(0^{2n})]_{2i - 1} = -2s_i^+ \neq 0\quad .$$
Thus by equation~\ref{eq:lowerbound}, the variables at positions 
$2k + 1$ and $2k + 2$ together have degree of at least $k$.
\end{proof}

\begin{cor}
Any VCSP implementing the winding fitness landscape from Definition~\ref{def:winding} must have a constraint graph that is dense and with unbounded treewidth.
\label{cor:unboundedTreewidth}
\end{cor}

\begin{proof}
Summing up over all $1 \leq k \leq n$, we get that any VCSP instance that implements $f$ must have total degree of at least $(n - 1)n/2$ (i.e. quadratic in the number $2n$ of variables).

In particular, this means that a constraint graph of bounded treewidth (which would have total degree linear in $2n$) cannot implement the winding fitness landscape $f$.
\end{proof}

In the following section we avoid the issue raised by Corollary~\ref{cor:unboundedTreewidth} by using an explicit construction based on a low-treewidth constraint graph.

\section{Exponential local search with bounded treewidth}
\label{sec:exponential}

We seek a simple model whose $n$ variables are all Boolean (with domain $\{0,1\}$), for which steepest ascent local search, by flipping the value of a variable at each step, has an exponentially long ascending sequence.  
In the model we construct, the soft constraints will have scopes of arity eight, namely the variables with index $4i,\dots,4i + 7$.  
The constructed model has a constraint graph of treewidth seven.  

In order to better explain the judgments made in constructing our novel fitness landscape, we build the model in a series of stages.  
At each stage we modify or extend the set of domains, variables and soft constraints from the previous stage so as to reach the final model which satisfies the above conditions.

\subsection{Stage 1: Counting}

Our initial model will have $N$ variables denoted $X_N,\dots,X_1$, which will all have the same finite domain including the values 0 and 1.

We will denote a state, i.e. an assignment of the $N$ variables, as an array of their values with the value of $X_1$ on the right.  
Thus the assignment $X_1= 1, X_2 = 0, X_3 = 0$ will be denoted by the array $\langle 0,0,1 \rangle$.
A local search path can thus be represented as a sequence of arrays.
The final exponential-length ascending path will include as a sub-sequence an encoding of the standard Boolean \textit{counting} sequence $0, 1, 10, 11, 100, 101,\dots $ or, more precisely:

\begin{equation}
0^N, 0^{N-1}1, 0^{N-2}10, 0^{N-2}11, 0^{N-3}100, 0^{N-3}101,\dots  
\end{equation}

\noindent which is highlighted in all of the explicit examples.

Since the final model will have $O(N)$ Boolean variables this will be a sub-sequence of exponential length.

\subsection{Stage 2: Two new domain values}

The values $0$ and $1$ alone do not enable us to write down a sequence of arrays of increasing value reached by local search.  
We only allow a local search step to flip a single variable value.
Since many of the steps in the binary counting sequence above alter the value of many more than one variable, they cannot be steps in our final ascent.

To fix this we introduce a \textit{carry} symbol, $C$.  
The idea is that carry represents the fact that we have arrived at $S01^i$, where $S$ is some arbitrary prefix of length $N-i-1$. 
We first flip the rightmost $1$ to be the new carry symbol.  We then propagate this carry symbol to the left so long as it is preceded by a  $1$. 
Each propagation first yields two adjacent carry symbols, and then we replace the right hand carry with a $0$.  Eventually we have only one carry symbol and it has a $0$ to its left.  
Now we flip this $0$, replacing $0C$ with a $1C$, and we arrive at $S1C0^{i-1}$.  
We can now flip the final $C$ into a $0$ and we have successfully performed a carry.
For example, to count from $[15]_2 = 1111$ to $[16]_2 = 10000$:

\begin{tabbing}
\hspace{2cm} \= \textbf{1:}\hspace{5mm} \= 0 \ \= 0 \ \= 0 \ \= 1 \ \= 1 \ \= 1 \ \= 1 \\
\>2:\ \> 0 \> 0 \> 0 \> 1 \> 1 \> 1 \> C \\
\>3:\ \> 0 \> 0 \> 0 \> 1 \> 1 \> C \> C \\
\>4:\ \> 0 \> 0 \> 0 \> 1 \> 1 \> C \> 0 \\
\>5:\ \> 0 \> 0 \> 0 \> 1 \> C \> C \> 0 \\
\>6:\ \> 0 \> 0 \> 0 \> 1 \> C \> 0 \> 0 \\
\>7:\ \> 0 \> 0 \> 0 \> C \> C \> 0 \> 0 \\
\>8:\ \> 0 \> 0 \> 0 \> C \> 0 \> 0 \> 0 \\
\>9:\ \> 0 \> 0 \> 1 \> C \> 0 \> 0 \> 0 \\
\>\textbf{10:}\ \> 0 \> 0 \> 1 \> 0 \> 0 \> 0 \> 0
\end{tabbing}

\noindent In this sequence only one variable is flipped at each step.
But note that the pair $1C$ changes in different ways.  
At lines 2, 4 and 6 $1C$ flips to $CC$.  However at line 9 $1C$ flips to $10$.
This is because the carry symbol appears in two contexts.  
We introduce $C$ to indicate a carry, but when its job is done and the $0$ on its left has become a $1$, the carry symbol needs to be removed. 

So far, the construction is equivalent to \manaut{Kaznatcheev, Cohen and Jeavons}\textcite{artemCP}'s Example 2 and with the right choice of soft-constraints is sufficient to show that some fitness increasing path has exponential length.
But this needs further elaboration to make the \emph{steepest} ascent exponentially long and reduce domain size to $2$.

To eliminate the dual role of carry as add and remove, we introduce a fourth symbol $X$ and put extra steps to enable the transition from state 8 to state 9.
When the $0C$ should become $1C$, we first move to $XC$, then remove the carry symbol, and then flip the $X$ to a $1$. 

\begin{tabbing}
\hspace{2cm} \= 8:\hspace{5mm}\= 0 \ \= 0 \ \= 0 \ \= C \= 0 \= 0 \= 0 \\
\> 9a: \> 0 \> 0 \> X \> C \> 0 \> 0 \> 0 \\
\> 9b:\> 0 \> 0 \> X \> 0 \> 0 \> 0 \> 0 \\
\> \textbf{10:}\> 0 \> 0 \> 1 \> 0 \> 0 \> 0 \> 0
\end{tabbing}

The admissible configurations of these symbols can be summarised in as the six cases in equation \ref{eq:01-X}:

\begin{equation}
\begin{aligned}
\{01\}+ & &
\{01\}^*1C\{01\}^*0 & &
\{01\}^*0C\{01\}^*0\\
\{01\}^*CC\{01\}^*0 & &
0^*XC\{01\}^*0 & &
0^*X\{01\}^*0
\end{aligned}
\label{eq:01-X}
\end{equation}

\noindent A VCSP will be designed so that only admissible configurations of these symbols will be reached by steepest ascent from $0^N$.

\subsection{Stage 3: New intermediate symbols}

Even with our new symbols $C$ and $X$, the sequence we proposed in the previous section cannot 
correspond to a steepest ascent.
Simply flipping the most significant variable from 0 to 1 (so $0^N$ becomes $10^{N-1}$) jumps directly from the start the the end of the sequence.

It is necessary to ensure that a state $s$ cannot be a neighbour of a previous state $s'$ (apart from its immediate predecessor in the path) otherwise the steepest-ascent would take the short-cut from $s'$ to $s$.
In order to constrain and prioritise moves we introduce an intermediate symbol between each pair of symbols.
It thus requires two steps to change any of the main symbols, $0,1,C$ and $X$ to another main symbol, via an intermediate symbol.

To encode all (4 main and 6 intermediate) symbols we will use four boolean variables.  The encoding, and the definition of the soft constraints, will ensure that a flip of a boolean variable which increases the objective 
(a) can only yield boolean sequences that encode (main or intermediate) symbols;
(b) cannot change a main to a main symbol;
(c) cannot change an intermediate to an intermediate symbol;
and, (d) can only change a main to an intermediate symbol if there currently are no intermediate symbols (i.e. all flips from intermediate symbols to main symbols are better than any move from a main symbol to an intermediate symbol)


Each intermediate symbol can be chosen to be one flip away from exactly two main symbols: 
it falls therefore ``between" the two main symbols.  
Consequently, when an intermediate symbol is flipped during a sequence of local moves, in order to improve again, it must move forward to the other main symbol.

We build the set of (six) intermediate symbols using a function symbol $i$ where $i_{ab} = i_{ba}$ is the intermediate symbol between main symbols $a$ and $b$.
We denote the set of symbols: 

\begin{equation}
\Sym = \{0,1,C,X,i_{01}, i_{C,0}, i_{0,X}, i_{1,C}, i_{X,1}, i_{C,X}\}
\end{equation}

\paragraph{Boolean encoding of main and intermediate symbols.}
The explicit encoding using the four Boolean variables $x_{1,i}, x_{2,i}, x_{3,i}, x_{4,i}$ is given in the following table.  

\[
\begin{array}{r| c c c c}
X_i & x_{1,i} & x_{2,i} & x_{3,i} & x_{4,i}\\
\hline
 0 & 1 & 0 & 0 & 0\\
 1 & 0 & 1 & 0 & 0 \\
 C & 0 & 0 & 1 & 0 \\
 X & 0 & 0 & 0 & 1 \\
i_{01}  & 1 & 1 & 0 & 0 \\         
i_{C0}  & 1 & 0 & 1 & 0 \\        
i_{0X}  & 1 & 0 & 0 & 1 \\        
i_{1C}  & 0 & 1 & 1 & 0 \\       
i_{X1}  & 0 & 1 & 0 & 1 \\       
i_{CX} & 0 & 0 & 1 & 1         
 \end{array}
 \]
 
Each intermediate symbol has two Boolean variables set to $1$, and each original symbol corresponds to a pattern with just one Boolean variable set to $1$.
The row $\langle 0,0,0,0 \rangle$ and any row with more than two $1$'s represents a non-symbol.

The intermediate symbols mean that the example of a steepest-ascent path will now take the following form (where each symbol would now be represented by four Boolean variables):

\begin{tabbing}
\hspace{2cm} \= \textbf{1:}\hspace{5mm} \= 0 \ \ \= 0 \ \ \= 0 \ \ \= 1 \ \ \= 1 \ \ \= 1 \ \ \= 1 \\
\> \> 0 \> 0 \> 0 \> 1 \> 1 \> 1 \> $i_{1C}$ \\ 
\> 2: \> 0 \> 0 \> 0 \> 1 \> 1 \> 1 \> C \\
\> \> 0 \> 0 \> 0 \> 1 \> 1 \> $i_{1C}$ \> C \\
\> 3: \> 0 \> 0 \> 0 \> 1 \> 1 \> C \> C \\
\> \> 0 \> 0 \> 0 \> 1 \> 1 \> C \> $i_{C0}$ \\
\> 4: \> 0 \> 0 \> 0 \> 1 \> 1 \> C \> 0 \\
\> \> 0 \> 0 \> 0 \> 1 \> $i_{1C}$ \> C \> 0 \\
\> 5: \> 0 \> 0 \> 0 \> 1 \> C \> C \> 0 \\
\> \> 0 \> 0 \> 0 \> 1 \> C \> $i_{C0}$ \> 0 \\
\> 6: \> 0 \> 0 \> 0 \> 1 \> C \> 0 \> 0 \\
\> \> 0 \> 0 \> 0 \> $i_{1C}$ \> C \> 0 \> 0 \\
\> 7: \> 0 \> 0 \> 0 \> C \> C \> 0 \> 0 \\
\> \> 0 \> 0 \> 0 \> C \> $i_{C0}$ \> 0 \> 0 \\
\> 8: \> 0 \> 0 \> 0 \> C \> 0 \> 0 \> 0 \\
\> \> 0 \> 0 \> $i_{0X}$ \> C \> 0 \> 0 \> 0 \\
\> 9a: \> 0 \> 0 \> X \> C \> 0 \> 0 \> 0 \\
\> \> 0 \> 0 \> X \> $i_{C0}$ \> 0 \> 0 \> 0 \\
\> 9b: \> 0 \> 0 \> X \> 0 \> 0 \> 0 \> 0 \\
\> \> 0 \> 0 \> $i_{X1}$ \> 0 \> 0 \> 0 \> 0 \\
\> \textbf{10:}\> 0 \> 0 \> 1 \> 0 \> 0 \> 0 \> 0
\end{tabbing}

There are now more possible configurations than those six that appear in equation~\ref{eq:01-X}.  
The additional admissible configurations that include intermediate symbols are $\{01\}^+i_{01}$, $\{01\}^+i_{1C}$, $\{01\}^+i_{1C}C\{01\}^*$, $0^*i_{0X}C\{01\}^*$, $\{01\}^+Ci_{C0}\{01\}^*$, $0^*Xi_{C0}\{01\}^*$.
The VCSP soft constraints will be designed so that only these ``intermediate" configurations will occur on the steepest ascent from $0^N$.

\subsection{Encoding using Soft Constraints}
\paragraph{Ensuring that intermediate symbols only occur in admissible configurations}
The VCSP which assigns a value to each state comprises just two constraints, a unary constraint $h$ which applies only to the last symbol, and a binary constraint $f$.  
The constraint $f$ applies to each pair of neighbouring symbols. 
By assigning a value of $0$ to each inadmissible pair of symbols it is simple to prevent an improving flip generating an intermediate symbol in an inadmissible configuration.

Indeed, for each intermediate symbol, we can list the non-zero rows of $f$ that include that symbol:
\begin{tabular}{llll}
     $i_{01}$: None & $\;\;$ & $i_{C0}$: $f(C,i_{C0}), f(X,i_{C0})$ &  $i_{0X}$: $f(i_{0X},C)$\\
     $i_{1C}$: $f(i_{1C},C)$ & $\;\;$ & $i_{X1}$: $f(i_{X1},0)$ &  $i_{CX}$:  None
\end{tabular}

\noindent Let us call the sequence (from $0^N$ to $01^{N-1}$) via steepest ascent, a counting path.
Each of these non-zero-valued rows lies between a predecessor and successor on a counting path:

\begin{tabbing}
    $\quad$ CC $\Rightarrow$ C$i_{C0}$ $\Rightarrow$ C0 $\quad$ XC $\Rightarrow$ X$i_{C0}$ $\Rightarrow$ X0\\
    $\quad$ 0C $\Rightarrow$ $i_{0X}$C $\Rightarrow$ XC $\quad$
    1C $\Rightarrow$ $i_{1C}$C $\Rightarrow$ CC\\
    $\quad$ X0 $\Rightarrow$ $i_{X1}0$ $\Rightarrow$ 10
\end{tabbing}

\noindent The values of $f$ for each row will be set so that both forward flips yield an increase in the valuation of the state.
Consequently a backward flip, from the resulting configurations \verb9C0,X0,XC,CC,109, will always result in a lower valuation, and therefore never occur.

Note, in particular that the pair of main symbols 
\verb9CC, XC, 0C, 1C, X09
which yield these configurations by a flip can only occur once in an admissible configuration.
Thus through these flips, two intermediate symbols will never appear in a configuration on a counting path.

To enable the symbols $i_{01}$ and $i_{1C}$ to occur as the last symbol, the soft constraint $h$, gives a non-zero value to $h(i_{01})$ and to $h(i_{1C})$.  
The values of $h$ are small enough that on a counting path $i_{01}$ and $i_{1C}$ will never occur in a configuration that has any other intermediate symbols.
The forward flips introduced above are always more valuable than $h$.

\paragraph{Ensuring the counting path only yields admissible configurations}
We show in this section that it is possible to make this counting path a steepest ascent by a judicious choice of soft constraints. 
We first identify a set of prioritised transition rules which if followed from the
initial tuple $0^N$ produce exactly the counting path. 
We then construct soft constraints which define a fitness landscape on which the transition rules correspond to an ascending path. 
The list of transition rules is given below.
\begin{enumerate}
\item Change the last symbol from 0 to 1 (incrementing the counting number)
\begin{enumerate}
    \item If $X_2 \in \{0,1\}$ and $X_1 = 0$  then set $X_1 = i_{01}$.\label{incr-rule}
        \item If $X_2 \in \{0,1\}$ and $X_1 = i_{01}$ , then set $X_1 = 1$.\label{incr-rule2}
\end{enumerate}  
\item Change the last symbol from 1 to C
\begin{enumerate}
    \item If $X_2 \in \{0,1\}$ and $X_1 = 1$  then set $X_1 = i_{1C}$.\label{firstCarry-rule}
    \item If $X_2 \in \{0,1\}$ and $X_1 = i_{1C}$ , then set $X_1 = C$.\label{firstCarry-rule2}
\end{enumerate}
\item Carry to 1
\begin{enumerate}
    \item If $X_{i+1} = 1$ and $X_i = C$ then set $X_{i+1} = i_{1C}$.\label{duplicateCarry-rule}
    \item If $X_{i+1} = i_{1C}$ and $X_i = C$ then set $X_{i+1} = C$.\label{duplicateCarry-rule2}
\end{enumerate}
\item Carry to 0 
\begin{enumerate}
    \item If $X_{i+1} = 0$ and $X_i = C$ then set $X_{i+1} = i_{0X}$.\label{temporise-rule}
    \item If $X_{i+1} = i_{0X}$ and $X_i = C$ then set $X_{i+1} = X$.\label{temporise-rule2}
\end{enumerate}
\item Drop the carry after C
\begin{enumerate}
    \item If $X_{i+1} = C$ and $X_i = C$ then set $X_i = i_{C0}$.\label{secondCarry-rule}
    \item If $X_{i+1} = C$ and $X_i = i_{C0}$ then set $X_i = 0$.\label{secondCarry-rule2}
\end{enumerate}
\item Drop the carry after X
\begin{enumerate}
    \item If $X_{i+1} = X$ and $X_i = C$ then  set $X_i = i_{C0}$.\label{dropCarry-rule}
    \item If $X_{i+1} = X$ and $X_i = i_{C0}$ then  set $X_i = 0$.\label{dropCarry-rule2}
\end{enumerate}
\item Use the final carry X
\begin{enumerate}
    \item If $X_{i+1} = X$ and $X_i = 0$ then  set $X_{i+1} = i_{X1}$.\label{actualCarry-rule}
    \item If $X_{i+1} = i_{X1}$ and $X_i = 0$ then  set $X_{i+1} = 1$.\label{actualCarry-rule2}
    \end{enumerate}
\end{enumerate}
\label{rules}

These rules are applicable in the following main symbol configurations:

\begin{align*}
&\{01\}^*0& Rule \ref{incr-rule}\\
&\{01\}^*1& Rule \ref{firstCarry-rule}\\
&\{01\}^*1C\{01\}^*0& Rule \ref{incr-rule}, Rule \ref{duplicateCarry-rule}\\
&\{01\}^*0C\{01\}^*0& Rule \ref{incr-rule}, Rule \ref{temporise-rule}\\ 
&\{01\}^*CC\{01\}^*0& Rule \ref{incr-rule}, Rule \ref{duplicateCarry-rule}, Rule \ref{temporise-rule}, Rule \ref{secondCarry-rule}\\
&0^*XC\{01\}^*0& Rule \ref{incr-rule}, Rule \ref{dropCarry-rule}\\
&0^*X\{01\}^*0& Rule \ref{incr-rule}, Rule \ref{actualCarry-rule}
\end{align*}

\noindent To maintain the counting property, Rule \ref{incr-rule} must only fire in the first configuration.
To keep to admissible configurations, Rule \ref{secondCarry-rule} must fire in the fifth configuration.

Thus the priorities we impose between these rules are as follows:
    (1) We prefer Rule~\ref{secondCarry-rule}, 
             over both Rules~\ref{duplicateCarry-rule}, \ref{temporise-rule}; and
    (2) We prefer any of Rules~\ref{duplicateCarry-rule}, \ref{temporise-rule}, 
             \ref{secondCarry-rule},
             \ref{dropCarry-rule},
             \ref{actualCarry-rule}, 
              over Rule~\ref{incr-rule}.
The VCSP will ensure that other flips do not improve on the current valuation.

The rules applicable in the states which have an intermediate value are:

\begin{align*}
&\{01\}^+i_{01}& Rule \ref{incr-rule2}\\
&\{01\}^+i_{1C}& Rule \ref{firstCarry-rule2}\\
&\{01\}^+i_{1C}C\{01\}^*& Rule \ref{duplicateCarry-rule2}\\
&0^*i_{0X}C& Rule \ref{temporise-rule2}\\
&0^*i_{0X}C\{01\}^*0& Rule  \ref{incr-rule}, Rule \ref{temporise-rule2}\\
&\{01\}^+Ci_{C0}& Rule \ref{duplicateCarry-rule}, Rule \ref{temporise-rule}, Rule \ref{secondCarry-rule2}\\
&\{01\}^+Ci_{C0}\{01\}^*0& Rule \ref{incr-rule}, Rule \ref{duplicateCarry-rule}, Rule \ref{temporise-rule}, Rule \ref{secondCarry-rule2}\\
&0^*Xi_{C0}\{01\}^*& Rule \ref{dropCarry-rule2}
\end{align*}

\noindent To ensure only admissible configurations Rule \ref{temporise-rule2} must fire in the fifth configuration, and Rule \ref{secondCarry-rule2} must fire in the sixth and seventh configurations.
By applying the above rules according to these priorities we ensure that only admissible states appear on the steepest ascent from $0^N$.  
We term this the Counting Path Property (CPP).
%
%
%
%
%
It is straightforward to verify the following CPP lemma:

\begin{lem}  
\label{lem:CPP}
All flips when applied to an admissible state yield an admissible state by applying the above prioritised transition rules.
\end{lem}

The following table introduces the binary soft constraint f(a,b) which we use to construct a fitness landscape in which the counting path is a steepest ascent. 
We obtained these values by solving a system of linear equations (described in the appendix) which ensure that a steepest-ascent respects the transition rules and their relative priorities. 
Many solutions exist, but to simplify the proof, we chose a solution with integer values and with a large number of zeros.

\[
\begin{array}{c|c|ccccccccccc}
\multicolumn{2}{c|}{f(a,b)}  & \multicolumn{10}{c}{b} \\
\cline{3-12} 
 \multicolumn{2}{c|}{}  & 0 & 1 & C & X & i_{01} & i_{C0} & i_{0X} & i_{1C} & i_{X1} & i_{CX} \\
\cline{1-12}
\multirow{10}{*}{a}  & 0   & 0 & 4 & 6 & 0 & 0 & 0 & 0 & 0 & 0 & 0 \\
                     & 1   & 0 & 4 & 6 & 0 & 0 & 0 & 0 & 0 & 0 & 0 \\
                     & C   &13 & 0 & 0 & 0 & 0 & 8 & 0 & 0 & 0 & 0 \\
                     & X   &13 & 0 & 8 & 0 & 0 &12 & 0 & 0 & 0 & 0 \\
                 & i_{01}  & 0 & 0 & 0 & 0 & 0 & 0 & 0 & 0 & 0 & 0 \\
                 & i_{C0}  & 0 & 0 & 0 & 0 & 0 & 0 & 0 & 0 & 0 & 0 \\
                 & i_{0X}  & 0 & 0 & 7 & 0 & 0 & 0 & 0 & 0 & 0 & 0 \\
                 & i_{1C}  & 0 & 0 &23 & 0 & 0 & 0 & 0 & 0 & 0 & 0 \\
                 & i_{X1}  &14 & 0 & 0 & 0 & 0 & 0 & 0 & 0 & 0 & 0 \\
                 & i_{CX}  & 0 & 0 & 0 & 0 & 0 & 0 & 0 & 0 & 0 & 0 \\
\end{array}
\]

We apply $f$ to each successive pair $X_{i+1},X_i$ of variables, with increasing weight $4^{i-1}$.
We also apply the cost function $h$, whose sole purpose is to trigger rules~\ref{incr-rule} and \ref{firstCarry-rule}, with weight 1 to $X_2,X_1$.
The function $h$ is identically zero except for the values $h(i_{01}) = 1$ and $h(i_{1C}) = 5$.
The value of a state is thus:

\begin{equation}
F(X) = h(X_1) + \sum_{i=1}^{N} 4^{i-1} f(X_{i+1}, X_i)
\end{equation}

It remains to be shown that a steepest-ascent algorithm on the fitness landscape defined by $F$, starting at the initial state $0^N$ respects the transition rules and their priorities. 
By Lemma~\ref{lem:CPP}, since $0^N$ clearly satisfies CPP, we only need to consider states satisfying CPP.
It suffices to show that: 
(a) each transition triggered by the rules leads to an increase in $F$,
(b) when conflicts arise the priorities between rules are respected by steepest ascent, 
and (c) no other transition leads to an increase in $F$.

\begin{lem} \label{lem:final}
In the fitness landscape defined by the global objective function $F$, the counting path is a steepest-ascent.
\end{lem}

A series of lemmas in appendix~\ref{app:LemmaProofs} contains a proof of lemma \ref{lem:final}, and we have also verified the steepest ascent by coding in Python and Eclipse~\cite{github-anon}.
We can now state our main theorem:

\begin{thm}
Steepest ascent may take an exponential number of steps to reach a local optimum even when the fitness landscape is defined by soft constraints over Boolean variables with constraint graph of pathwidth and treewidth 7. 
\end{thm}

\begin{proof}
Since the counting path is of exponential length, Lemma~\ref{lem:final} tells us that steepest-ascent on the landscape defined by $F$ requires an exponential number of steps to reach a local optimum starting from the initial state $0^N$. 
The constraint graph of the $X_i$ variables is clearly a chain. 
When we replace each $X_i$ variable by the Boolean variables $x_{1,i}, x_{2,i}, x_{3,i}, x_{4,i}$, the constraint graph has treewidth $7$. 
To see this, observe that when the variables are ordered in lexicographic order ($x_{a,i} < x_{b,j}$ if $i<j$ or ($i=j$ and $a<b$)), the only earlier variables that constraint a variable are among the 7 variables that immediately precede it in this order. 
If we consider the preceeding constraining variables along with the variable they constraint as an interval then we can see that our constraint graph is an interval graph with maximum clique size $8$ -- so it has pathwidth $7$ and thus treewidth $7$.
\end{proof}

\section{Conclusion and Discussion}

We have shown that steepest-ascent local search can require an exponential number of steps to reach a local maximum even when domains are boolean and the constraint graph has treewidth (in fact, the more restrictive pathwidth) bounded by a constant $k$. 
It is an open question to determine the smallest value of $k$ for which this holds. 
We have given a proof for $k=7$. 
\manaut{Kaznatcheev, Cohen, and Jeavons}\textcite{artemCP} have shown that for $k=1$ all ascending paths are bounded by a quadratic number of steps and that for $k=2$ there exist some ascending paths that are exponentially long.
But their long ascending paths are not the paths that steepest ascent would follow,
so for values of $k$ between 2 and 6 the question is still open for steepest ascent.

The fact that local search may require exponential time to converge has implications both in the performance analysis of local search algorithms and in the understanding of biological processes such as evolution.
Of course, we are aware that the example we have constructed is pathological. 
For example, experiments on computational protein design benchmark problems indicate that the average number of steps required by steepest ascent to reach a local optimum is proportional to the logarithm of the size of the attraction basin of the local optimum~\cite{Simoncini}, and hence (sub)linear in the total number of variables.
Further theoretical research is required to identify conditions under which exponentially-long steepest-ascent paths cannot occur.

\balance
\section*{Acknowledgements}
We thank Chris Watkins for spotting a discrepancy between the table for $f(a,b)$ and verification code in an earlier draft.
This project has received funding from the European Research 
Council (ERC) under the European
Union's Horizon 2020 research and innovation programme 
(grant agreement No 714532). The paper reflects only the 
authors' views and not the views of the ERC
or the European Commission. The European Union is not liable 
for any use that may be made of the information contained therein.
David Cohen was supported by Leverhulme Trust Grant RPG-2018-161.


\newpage
\newgeometry{textwidth=13.5cm,textheight=19.75cm} 
\onecolumn
\appendix
\section{Proofs of Lemmas}
\label{app:LemmaProofs}
\begin{lem}  \label{lem:increasing}
If the CPP is satisfied, then all transition rules lead to an increase in the objective function $F$.
\end{lem}

\begin{proof}
In the following let $S,S'$ be any (possibly empty) string and let $a$ be any symbol from $\{0,1\}$.
To show that applying Rules~\ref{incr-rule}, \ref{firstCarry-rule}, \ref{incr-rule2} or \ref{firstCarry-rule2} leads to an increase in $F$, we only need to show that

\begin{equation*}
F(S \ a\ 0) < F(S\ a\  i_{01}) < F(S\ a\ 1) < F(S\ a\  i_{1C}) < F(S\ a\ C)
\end{equation*}

\noindent This holds since
\begin{eqnarray*}
f(a,0) + h(0) = 0 & < \  f(a,i_{01}) + h(i_{01}) = 1 & < \  f(a,1) + h(1) = 4 \\
& < \  f(a,i_{1C}) + h(i_{1C}) = 5 & <  \ f(a,C) + h(C) = 6
\end{eqnarray*}

\noindent For Rules~\ref{duplicateCarry-rule} and \ref{duplicateCarry-rule2} to increase $F$, we require
\[ F(S\ a\ 1\ C\ S') \ < \ F(S\ a\ i_{1C}\ C\ S') \ < \ F(S\ a\ C\ C\ S')
\]
(where $a \in \{0,1\}$ follows from the CPP). These inequalities hold since
\begin{eqnarray*}
& 4f(a,1) + f(1,C) = 22 \ < \ 4f(a,i_{1C}) + f(i_{1C},C) = 23 \ < \ 4f(a,C) + f(C,C) = 24 &
\end{eqnarray*}
For Rules~\ref{temporise-rule} and \ref{temporise-rule2} to increase $F$, we require
\[ F(S\ a\ 0\ C\ S') \ < \ F(S\ a\ i_{0X}\ C\ S') \ < \ F(S\ a\ X\ C\ S')
\]
(where $a \in \{0,1\}$ follows from the CPP). These inequalities hold since
\begin{eqnarray*}
& 4f(a,0) + f(0,C) = 6 \ < \ 4f(a,i_{0X}) + f(i_{0X},C) = 7 \ < \ 4f(a,X) + f(X,C) = 8 &
\end{eqnarray*}
For Rules~\ref{dropCarry-rule} and \ref{dropCarry-rule2} to increase $F$, we require
\[ F(S\ X\ C\ S') \ < \ F(S\ X\ i_{C0}\ S') \ < \ F(S\ X\ 0\ S')
\]
where, by the CPP, $S'$ is a string of zeros. This holds since
\begin{eqnarray*}
& 4f(X,C) + f(C,0)  = 45 \ < \ 4f(X,i_{C0}) + f(i_{C0},0) = 48 \ < \ 52 = 4f(X,0) + f(0,0)  &
\end{eqnarray*}
when $S'$ is non-empty (and $f(X,C) = 8 < f(X,i_{C0}) = 12 < f(X,0) = 13$
if $S'$ is the empty string).
For Rules~\ref{secondCarry-rule} and \ref{secondCarry-rule2} to increase $F$, we require
\[ F(S\ C\ C\ S') \ < \ F(S\ C\ i_{C0}\ S') \ < \ F(S\ C\ 0\ S')
\]
where, again by the CPP, $S'$ is a string of zeros. This holds since
\begin{eqnarray*}
& 4f(C,C) + f(C,0) = 13 \ < \ 4f(C,i_{C0}) + f(i_{C0},0) = 32 \ < \ 52 = 4f(C,0) + f(0,0) &
\end{eqnarray*}
when $S'$ is non-empty (and $f(C,C) =0 < f(C,i_{C0}) = 8 < f(C,0) =13$
if $S'$ is the empty string).
For Rules~\ref{actualCarry-rule} and \ref{actualCarry-rule2} to increase $F$, we require
\[ F(S\ a\ X\ 0 \ S') \ < \ F(S\ a\ i_{X1}\ 0 \ S') \ < \ F(S\ a\ 1 \ 0\ S')
\]
(where $a \in \{0,1\}$ and $S'$ is a string of zeros by the CPP). These inequalities hold since
\begin{eqnarray*}
&  4f(a,X) + f(X,0) = 13 \ < \ 4f(a,i_{X1}) + f(i_{X1},0) = 14 \ < \ 4f(a,1) + f(1,0) \leq 16 &
\end{eqnarray*}
\end{proof}

\begin{lem} \label{lem:346}
Assuming the CPP, when one of Rules~\ref{duplicateCarry-rule} or \ref{temporise-rule} are in conflict with one of 
Rules~\ref{secondCarry-rule} or \ref{secondCarry-rule2},
the transitions triggered by the latter rules lead to a greater increase in $F$.
\end{lem}

\begin{proof}
Given the limited form of states described by the CPP, there is only one type of conflict for each pair of rules.
In the following calcualtions, $S,S'$ represent (possibly empty) strings and (since we are requiring the CPP) the symbol $a$ will always stand for either 0 or 1.

Rules~\ref{duplicateCarry-rule} and \ref{secondCarry-rule} are only in conflict in the state $Sa1CCS'$.
For the priority of Rule~\ref{secondCarry-rule} to be respected, it suffices to show that 
\[ F(S \ a \ i_{1C} \ C \ C \ S') \ < \ F(S \ a \ 1 \ C \ i_{C0} \ S')
\]
where $S'$ is a string of zeros. This holds since
\begin{eqnarray*}
 & & 64f(a,i_{1C}) + 16f(i_{1C},C) + 4f(C,C) + f(C,0) = 381  \\ 
 &  & \ \ <  \ 384 = 64f(a,1) + 16f(1,C) + 4f(C,i_{C0}) + f(i_{C0},0)  
\end{eqnarray*}
when $S'$ is non-empty (and 
$16f(a,i_{1C}) + 4f(i_{1C},C) + f(C,C) = 92 < 96 = 16f(a,1) + 4f(1,C) + f(C,i_{C0})$
if $S'$ is the empty string).

Rules~\ref{duplicateCarry-rule} and \ref{secondCarry-rule2} are only in conflict in the state $Sa1Ci_{C0}S'$.
For the priority of Rule~\ref{secondCarry-rule2} to be respected, it suffices to show that 
\[ F(S \ a \ i_{1C} \ C \ i_{C0} \ S') \ < \ F(S \ a \ 1 \ C \ 0 \ S')
\]
where, by the CPP, $S'$ is a string of zeros. This holds since
\begin{eqnarray*}
 & & 64f(a,i_{1C}) + 16f(i_{1C},C) + 4f(C,i_{C0}) + f(i_{C0},0) = 400  \\ 
 &  & \ \ <  \ 404 = 64f(a,1) + 16f(1,C) + 4f(C,0) + f(0,0)  
\end{eqnarray*}
when $S'$ is non-empty (and 
$16f(a,i_{1C}) + 4f(i_{1C},C) + f(C,i_{C0}) = 100 <  101 = 16f(a,1) + 4f(1,C) + f(C,0)$
if $S'$ is the empty string).

Rules~\ref{temporise-rule} and \ref{secondCarry-rule} are only in conflict in the state $Sa0CCS'$.
For the priority of Rule~\ref{secondCarry-rule} to be respected, it suffices to show that 
\[ F(S \ a \ i_{0X} \ C \ C \ S') \ < \ F(S \ a \ 0 \ C \ i_{C0} \ S')
\]
where, by the CPP, $S'$ is a string of zeros. This holds since
\begin{eqnarray*}
 & & 64f(a,i_{0X}) + 16f(i_{0X},C) + 4f(C,C) + f(C,0) = 125  \\ 
 &  & \ \ <  \ 128 = 64f(a,0) + 16f(0,C) + 4f(C,i_{C0}) + f(i_{C0},0)  
\end{eqnarray*}
when $S'$ is non-empty (and 
$16f(a,i_{0X}) + 4f(i_{0X},C) + f(C,C) = 28 <  32 = 16f(a,0) + 4f(0,C) + f(C,i_{C0})$
if $S'$ is the empty string).

Rules~\ref{temporise-rule} and \ref{secondCarry-rule2} are only in conflict in the state $Sa0Ci_{C0}S'$.
For the priority of Rule~\ref{secondCarry-rule2} to be respected, it suffices to show that 
\[ F(S \ a \ i_{0X} \ C \ i_{C0} \ S') \ < \ F(S \ a \ 0 \ C \ 0 \ S')
\]
where, by the CPP, $S'$ is a string of zeros. This holds since
\begin{eqnarray*}
 & & 64f(a,i_{0X}) + 16f(i_{0X},C) + 4f(C,i_{C0}) + f(i_{C0},0) = 144  \\ 
 &  & \ \ <  \ 148 = 64f(a,0) + 16f(0,C) + 4f(C,0) + f(0,0)  
\end{eqnarray*}
when $S'$ is non-empty (and 
$16f(a,i_{0X}) + 4f(i_{0X},C) + f(C,i_{C0}) = 36 <  37 = 16f(a,0) + 4f(0,C) + f(C,0)$
if $S'$ is the empty string).
\end{proof}

\begin{lem}  \label{lem:1all}
Assuming the CPP, when Rule~\ref{incr-rule} is in conflict with one of the other rules,
the transition triggered by the other rule always leads to a greater increase in $F$ then Rule~\ref{incr-rule}.
\end{lem}

\begin{proof}
By the CPP, conflicts between Rule~\ref{incr-rule} and another rule 
can only occur in a state ending with a string of at least two zeros,
namely one of $S0C0^r$, $S1C0^r$, $SX0^r$, $Si_{X1}0^r$,  $S0CC0^r$, $S1CC0^r$, $SXC0^r$, 
$S0Ci_{C0}0^r$, $S1Ci_{C0}0^r$, $Si_{1C}C0^r$, $Si_{0X}C0^r$, $SXi_{C0}0^r$
for some string $S$ of zeros and ones, and with $r \geq 2$.
In each case, the application of Rule~\ref{incr-rule} leads to an increase of just 1 in the value of $F$,
whereas the application of any other rule leads to an increase of at least $4^{r-1} > 1$ (by the same 
exhaustive case analysis as given in the proof of Lemma~\ref{lem:increasing}).
\end{proof}

\begin{lem}  \label{lem:noother}
Assuming the CPP is satisfied, the only transitions which
lead to an increase in $F$ are those triggered by Rules 1-14.
\end{lem}

\begin{proof}
Recall that the only possible transitions are those between a main symbol ($0$, $1$, $C$ or $X$)
and an intermediate symbol, since all other transitions change 2 bits. 

We first consider all possible transitions from a main symbol to one of the six intermediate symbols:
\begin{itemize}
\item[{\bf $i_{C0}$: }] since the only non-zero values of $f$ (or $h$) with an argument $i_{C0}$ are
$f(X,i_{C0}) = 12$ and $f(C,i_{C0}) = 8$, the only transitions we need to consider are
$SX0S' \rightarrow SXi_{C0}S'$, $SC0S' \rightarrow SCi_{C0}S'$ and $SXCS' \rightarrow SXi_{C0}S'$, 
$SCCS' \rightarrow SCi_{C0}S'$. The first two transitions lead to a decrease in $F$ since they
correspond to the inverse transitions of Rules~\ref{dropCarry-rule2} and \ref{secondCarry-rule2}. 
The latter two transition correspond to Rules~\ref{dropCarry-rule} and \ref{secondCarry-rule}.
\item[{\bf $i_{0X}$: }] since the only non-zero value of $f$ (or $h$) with an argument $i_{0X}$
is $f(i_{0X},C) = 7$, the only transitions we need to consider are $S0CS' \rightarrow Si_{0X}XS'$
and $SXCS' \rightarrow Si_{0X}XS'$. The former corresponds to Rule~\ref{temporise-rule},
whereas the latter leads to a decrease in $F$ since it is the inverse transition of Rule~\ref{temporise-rule2}.
\item[{\bf $i_{1C}$: }] since the only non-zero values of $f$ or $h$ with an argument $i_{1C}$ are
$f(i_{1C},C) = 23$ and $h(i_{1C}) = 5$, the only transitions we need to consider are
$S1CS' \rightarrow Si_{1C}CS'$, $SCCS' \rightarrow Si_{1C}CS'$, and $Sa1 \rightarrow Sai_{1C}$ 
$SaC \rightarrow Sai_{1C}$ (where $a \in \{0,1\}$). Of these transitions, the first corresponds to 
Rule~\ref{duplicateCarry-rule}, the second to the inverse of Rule~\ref{duplicateCarry-rule2},
the third to Rule~\ref{firstCarry-rule} and the last to the inverse of Rule~\ref{firstCarry-rule2}.
\item[{\bf $i_{X1}$: }] since the only non-zero value of $f$ (or $h$) with an argument $i_{X1}$ is
$f(i_{X1},0) = 14$, the only transitions we need to consider are $S10S' \rightarrow Si_{X1}0S'$
and $SX0S' \rightarrow Si_{X1}0S'$. The former corresponds to the inverse of Rule~\ref{actualCarry-rule2}
and the latter corresponds to Rule~\ref{actualCarry-rule}.
\item[{\bf $i_{01}$: }] since the only non-zero values of ($f$ or) $h$ with an argument $i_{01}$ 
are $h(i_{01}) = 1$, the only transitions we need to consider are 
$Sa0 \rightarrow Sai_{01}$ and $Sa1 \rightarrow Sai_{01}$ (where $a \in \{0,1\}$).
The former corresponds to Rule~\ref{incr-rule} and the latter corresponds to the inverse of
Rule~\ref{incr-rule2}.
\end{itemize}

We now consider all possible transitions from one of the six intermediate symbols to a main symbol:
\begin{itemize}
\item[{\bf $i_{C0}$: }] since, by the CPP, $i_{C0}$ can only occur just after $X$ or $C$,
the only transitions we need to consider are $SXi_{C0}S' \rightarrow SX0S'$, $SCi_{C0}S' \rightarrow SC0S'$,
$SXi_{C0}S' \rightarrow SXCS'$, and $SCi_{C0}S' \rightarrow SCCS'$. Of these transitions, the first
corresponds to Rule~\ref{dropCarry-rule2}, the second to Rule~\ref{secondCarry-rule2}, the third
to the inverse of Rule~\ref{dropCarry-rule} and the last to the inverse of Rule~\ref{secondCarry-rule}.
\item[{\bf $i_{0X}$: }] since, by the CPP, $i_{0X}$ can only occur just before $C$, the only transitions
we need to consider are $Si_{0X}CS' \rightarrow S0CS'$ and $Si_{0X}CS' \rightarrow SXCS'$. The first
of these transitions corresponds to Rule~\ref{temporise-rule} and the second corresponds to the inverse
of Rule~\ref{temporise-rule2}.
\item[{\bf $i_{1C}$: }] since, by the CPP, $i_{1C}$ can only occur at $X_1$ (with $X_2 \in \{0,1\}$) or just before $C$,
the only transitions we need to consider are $Sai_{1C} \rightarrow Sa1$, 
$Sai_{1C} \rightarrow SaC$ (where $a \in \{0,1\}$), and $Si_{1C}CS' \rightarrow S1CS'$, $Si_{1C}CS' \rightarrow SCCS'$.
Of these transitions, the first is the inverse of Rule~\ref{incr-rule}, the second is Rule~\ref{firstCarry-rule2},
the third is the inverse of Rule~\ref{duplicateCarry-rule} and the fourth is Rule~\ref{duplicateCarry-rule2}.
\item[{\bf $i_{X1}$: }] since, by the CPP, $i_{X1}$ can only occur just before $0$, the only transitions we
need to consider are $Si_{X1}0S' \rightarrow S10S'$ and $Si_{X1}0S' \rightarrow SX0S'$. The first of
these transitions corresponds to Rule~\ref{actualCarry-rule2} and the second to the inverse of
Rule~\ref{actualCarry-rule}.
\item[{\bf $i_{01}$: }] since, by the CPP, $i_{01}$ can only occur at $X_1$ (with $X_2 \in \{0,1\}$), the
only transitions we need to consider are $Sai_{01} \rightarrow Sa0$ and $Sai_{01} \rightarrow Sa1$.
The former corresponds to the inverse of Rule~\ref{incr-rule} and the latter to Rule~\ref{incr-rule2}.
\end{itemize}
\end{proof}

Lemma \ref{lem:final} now follows directly from Lemmas~\ref{lem:increasing}, \ref{lem:346}, \ref{lem:1all}, \ref{lem:noother}, and the steepest ascent has been verified by coding in Python and Eclipse~\cite{github-anon}.

\newpage
\printbibliography

\end{document}